\documentclass[american,aps,pra,reprint,floatfix,nofootinbib,superscriptaddress,longbibliography]{revtex4-1}
\usepackage[unicode=true,pdfusetitle, bookmarks=true,bookmarksnumbered=false,bookmarksopen=false, breaklinks=false,pdfborder={0 0 0},backref=false,colorlinks=false]{hyperref}
\hypersetup{colorlinks,linkcolor=myurlcolor,citecolor=myurlcolor,urlcolor=myurlcolor}
\usepackage{graphics,epstopdf,graphicx,amsthm,amsmath,amssymb,mathptmx,braket,colortbl,color,bm,framed,mathrsfs}
\usepackage[T1]{fontenc}
\usepackage[up]{subfigure}
\usepackage{tikz}

\definecolor{myurlcolor}{rgb}{0,0,0.9}

\newcommand{\proj}[1]{| #1\rangle\!\langle #1 |}

\DeclareMathOperator{\trace}{Tr}

\newcommand{\Ptr}[2]{\trace_{#1}\Pa{#2}}
\newcommand{\Tr}[1]{\Ptr{}{#1}}

\newcommand{\Pa}[1]{\left[#1\right]}

\newcommand{\norm}[1]{\left\lVert #1 \right\rVert}

\newcommand{\Conv}{\mathrm{Conv}}

\theoremstyle{plain}
\newtheorem{thm}{Theorem}

\newtheorem{prop}[thm]{Proposition}
\newtheorem{cor}[thm]{Corollary}
\theoremstyle{definition}
\newtheorem{definition}[thm]{Definition}

\newcommand*{\myproofname}{Proof}

\def\ot{\otimes}

\DeclareMathAlphabet{\mathcal}{OMS}{cmsy}{m}{n}

\makeatother

\begin{document}

  \author{Kaifeng Bu}
 \email{kfbu@fas.harvard.edu}
\affiliation{Department of Physics, Harvard University, Cambridge, Massachusetts 02138, USA}

\author{Dax Enshan Koh}
  \email{dax\_koh@ihpc.a-star.edu.sg}
\affiliation{Institute of High Performance Computing, Agency for Science, Technology and Research (A*STAR), 1 Fusionopolis Way, \#16-16 Connexis, Singapore 138632, Singapore}

\author{Lu Li}
\affiliation{Department of Mathematics, Zhejiang Sci-Tech University, Hangzhou, Zhejiang 310018, China}
\affiliation{School of Mathematical Sciences, Zhejiang University, Hangzhou, Zhejiang 310027, China}

\author{Qingxian Luo}
\affiliation{School of Mathematical Sciences, Zhejiang University, Hangzhou, Zhejiang 310027, China}
\affiliation{Center for Data Science, Zhejiang University, Hangzhou, Zhejiang 310027, China}

\author{Yaobo Zhang}
\affiliation{Zhejiang Institute of Modern Physics, Zhejiang University, Hangzhou, Zhejiang 310027, China}
\affiliation{Department of Physics, Zhejiang University, Hangzhou, Zhejiang 310027, China}

\title{Rademacher complexity of noisy quantum circuits}

\begin{abstract}

Noise in quantum systems is a major obstacle to implementing many quantum algorithms on large quantum circuits. In this work, we study the effects of noise on the Rademacher complexity of quantum circuits, which is a measure of statistical complexity that quantifies the richness of classes of functions generated by these circuits. We consider noise models that are represented by convex combinations of unitary channels and provide both upper and lower bounds for the Rademacher complexities of quantum circuits characterized by these noise models. In particular, we find a lower bound for the Rademacher complexity of noisy quantum circuits that depends on the Rademacher complexity of the corresponding noiseless quantum circuit as well as the free robustness of the circuit. Our results show that the Rademacher complexity of quantum circuits decreases with the increase in noise.

\end{abstract}

\maketitle

\section{Introduction}

The last few years have seen the burgeoning of two different activities in quantum computing: algorithmic developments in quantum-enhanced machine learning \cite{BiamonteNature17,dunjko2018machine,dunjko2016quantum} and experimental developments in building noisy intermediate-scale quantum (NISQ) computers \cite{arute2019quantum,arute2020hartree,gong2021quantum,figgatt2017complete,gibney2020alternative}. In the former activity, a central goal is to design quantum algorithms for machine learning tasks that provide a substantial improvement in performance over classical algorithms. In the latter activity, a key target is to build quantum hardware that would allow for these quantum algorithms as well as algorithms for other applications to be implemented.

While rapid advancements have been made in the above activities, several fundamental challenges remain and need to be addressed. On the one hand, questions remain about how much of an advantage quantum machine learning models can provide over their classical counterparts and what the theoretical limitations of these models are \cite{wiebe2020key}. On the other hand, near-term quantum computers, albeit impressive, are too noisy and error-prone to run many quantum algorithms on large-enough input sizes with sufficiently small error \cite{preskill2018quantum}.

A number of works have sought to shed light on the first of these challenges (see, e.g., \cite{caro2020pseudo,huang2021information, cheng2016learnability, huang2020power,rocchetto2018stabiliser,abbas2020power,wright2020capacity} 
and also a recent survey on quantum learning theory \cite{srinivasan2017guest}). For example, in \cite{huang2021information}, Huang, Kueng, and Preskill study the complexity of training both classical and quantum machine learning models for predicting outcomes of experiments and show that quantum models can provide an exponential advantage over classical models for certain tasks. In \cite{caro2020pseudo}, Caro and Datta use the pseudo-dimension to characterize the expressive power of quantum circuits, which can be applied to bounding the gate complexity of quantum state preparation and the learnability of quantum circuits. In \cite{cheng2016learnability}, Cheng, Hsieh, and Yeh use the fat-shattering dimension to characterize the learnability of unknown quantum measurements and states.

In a series of papers \cite{bu2021on,bu2021effects}, we sought to further address this challenge by studying quantum circuits in terms of their Rademacher complexity, which is a notion of statistical complexity introduced in \cite{Bartlett03} that measures the richness of classes of real-valued functions and provides 
bounds on the generalization error \cite{Bartlett03,koltchinskii2006} associated with learning from training data. These bounds may in turn be used to determine how a hypothesis function may perform on unseen data drawn from an unknown probability distribution. In \cite{bu2021on}, we analyzed the dependence of the Rademacher complexity on various structure parameters of quantum circuits---like their depth and sizes of their input and output registers---and on a particular resource measure, namely the resource measure of magic \cite{howard2017application, wang2019quantifying}, quantified using the so-called $(p,q)$ group norm. In \cite{bu2021effects}, we extended the above analysis by studying the dependence of the Rademacher complexity on general quantum resources by employing tools from the framework of quantum resource theories \cite{COECKE16,chitambar_2019}.

In the above works, our analysis of the Rademacher complexity had been restricted to quantum circuits without noise. However, as the second of the aforementioned challenges highlights, 
this is an unrealistic restriction, especially in this pre-quantum-error-correction NISQ era wherein real-world quantum circuits are susceptible to noise. To understand the effects of noise on the power of quantum circuits, numerous recent works have studied how noise increases the classical simulability of quantum circuits  \cite{gao2018efficient,takahashi2020classically,fujii2016computational,bremner2017achieving,bu2019efficient}. This behooves us to ask a similar question about the Rademacher complexity: how does the Rademacher complexity of quantum circuits change as they are subject to noise?

In this paper, we answer the above question by investigating the effects of noise on the Rademacher complexity of quantum circuits.
More specifically, we consider noise models that are described by mixed-unitary channels, i.e.~quantum channels that can be decomposed as a convex combination of unitary channels. These include well-known noise channels such as the depolarizing channel and the dephasing channel \cite{nielsen2010quantum,watrous2018theory}. For noisy quantum circuits described by these noise models, we give
both upper and lower bounds for their Rademacher complexities. 
Our lower bound for the Rademacher complexity of the noisy quantum circuit depends on that of the corresponding noiseless quantum circuit as well as the free robustness \cite{bu2021effects} of the circuit. In addition, we show that the Rademacher complexity of a quantum circuit decreases as the amount of noise increases.

\section{Methods}

We follow a setup similar to that of \cite{bu2021on,bu2021effects}. Consider a sample $\vec{x}$, encoded as a quantum state
$\ket{\psi(\vec{x})}$. 
Applying a quantum circuit $C$ to it and then measuring a Hermitian operator $H$ with respect to the output produces an expected measurement outcome of
\begin{eqnarray}
f_C(\vec{x})=
\Tr{C(\proj{\psi(\vec x)})H},
\end{eqnarray}
which defines the real-valued function $f_C$. Given a set of 
quantum circuits $\mathcal{C}$, we define the
function class $\mathcal{F}\circ \mathcal{C}:=\set{f_C:C\in\mathcal{C}}$. The Rademacher complexity of $\mathcal{F}\circ\mathcal{C}$ 
 on $m$ independent samples $S=\set{\vec{x_1},\ldots,\vec{x}_m}$ is defined as
 \begin{eqnarray}
 R_S(\mathcal{F}\circ\mathcal{C})
 =\mathbb{E}_{\vec{\epsilon}}
\sup_{C\in\mathcal{C} } \frac{1}{m}
\left|\sum_i\epsilon_i f_{C}(\vec{x}_i)\right|,
 \end{eqnarray}
where each $\epsilon_i$ in the expectation above is a Rademacher random variable taking values $\pm 1$ with equal probability $1/2$.

To incorporate noise into our model, we will consider local noise channels $D$
that are allowed to act on single or multiple registers of the circuit. For a circuit $C$, such noise may be modeled by having noise channels replace bare wires in the $C$. We will further assume that each local noise channel $D$ is a mixed-unitary channel, i.e.~a convex combination of unitary channels. 
In other words, $D$ may be written as
\begin{align}
\label{eq:noichanel}
D(\rho)=
\left(1-\sum_{i}p_i\right)
\rho+\sum_ip_iU_i\rho U^\dag_i,
\end{align}
where $U_i$ is a unitary operator for all $i$ and $(p_i)_i$ is a stochastic vector, i.e.~each $p_i$ satisfies $0 \leq p_i \leq 1$ and $\sum_i p_i = 1$. While the results that we will derive would hold generally for all such $p_i$, it will be adequate for us to further assume that $\sum_ip_i\leq 1/2$--- i.e.~the noise channel $D$ has a probability of at least $1/2$ of transforming the state correctly---as the results we derive (e.g.~Eq.~\eqref{eq:lowerbound}) are meaningful only when this holds. 

Note that for qubit systems, the channels described by Eq.~\eqref{eq:noichanel} are precisely the unital channels \cite{mendl2009unital,tregub1986bistochastic}, i.e.~channels for which the identity operator is a fixed point. Such channels include the well-known depolarizing channel and dephasing channel.

Next, we introduce the following definition.

\begin{definition}[$\mathcal C$-compatible]
Let $\mathcal C$ be a set of quantum circuits and $D$ be a mixed-unitary quantum channel. 
We say that $D$ is $\mathcal C$-\textit{compatible} if for any circuit $C \in \mathcal C$ and for any set of wires $W$ in $C$, there exists a set of unitaries $\mathcal U = \{U_i\}_i$ such that
\begin{enumerate}
    \item $D$ is a convex combination of $U_i(\cdot) U_i^\dag$'s, where $U_i \in \mathcal U$, and the identity superoperator (i.e.~$D$ is of the form given by Eq.~\eqref{eq:noichanel}), and
    \item for all $i$, the circuit $\tilde C_i$ formed by replacing the wires in $W$ by $U_i$ satisfies $\tilde C_i \in \mathcal C$.
\end{enumerate}
\end{definition}

 \begin{figure}[h]
  \center{\includegraphics[width=8cm]  {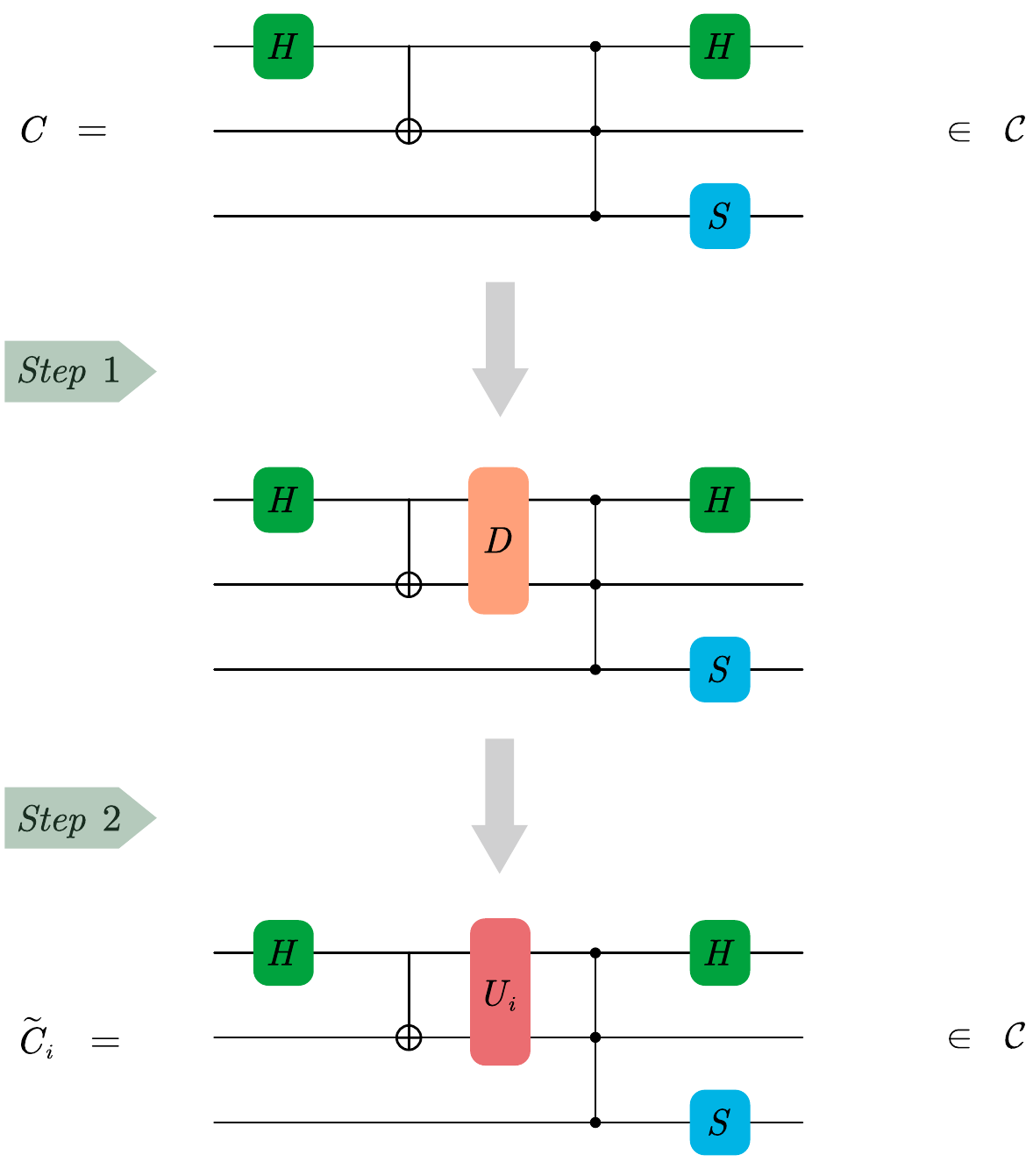}}     
  \caption{An illustration describing the definition of $\mathcal C$-compatibility of mixed-unitary quantum channels.}
  \label{fig0}
 \end{figure}

Intuitively, this may be understood as follows. Take any circuit $C \in \mathcal C$.  Consider a noisy version $\tilde C$ of $C$ that is formed by having the local noise channel $D$ act on one of the wires in $C$ at a specified location. Now, if there exists a decomposition \eqref{eq:noichanel} of $D$ such that for any $i$, replacing the channel $D$ in $\tilde C$ by the gate $U_i$ results in a circuit that is still in $\mathcal C$, then we say that $D$ is $\mathcal C$-\textit{compatible} (See Fig.~\ref{fig0} for an illustration describing this definition). 
For example, let $\mathcal C_n$ denote the set of Clifford circuits and let 
\begin{align}\label{eq:depC}
D_{\epsilon}(\rho)=(1-3\epsilon)\rho
+\epsilon X\rho X+\epsilon Y\rho Y+\epsilon Z\rho Z
\end{align}
denote the depolarizing channel. Since each of the Pauli matrices $X, Y, Z$ is in the Clifford group, the depolarizing channel $D_{\epsilon}$ is $\mathcal C_n$-compatible.

\section{Results}

\subsection{Upper bound}

In this section, we will prove some upper bounds on the Rademacher complexity of noisy circuits. Consider a collection $\mathcal D$ of local noise channels. Let $\mathcal{C}^{(k)}_{\text{noisy}}$ denote the set 
of quantum circuits with exactly  $k$ occurrences of noise channels from $\mathcal D$ acting on some positions of an ideal quantum circuit $C \in \mathcal{C}$.  With this definition, note that we have 
$\mathcal{C}^{(0)}_{\text{noisy}}=\mathcal{C}$.
For each noisy circuit class $\mathcal{C}^{(k)}_{\text{noisy}}$, we obtain 
the corresponding function class
$\mathcal{F}\circ\mathcal{C}^{(k)}_{\text{noisy}}$. We now state and prove a relationship between the 
Rademacher complexities  of $\mathcal{F}\circ\mathcal{C}^{(k)}_{\text{noisy}}$ and 
$\mathcal{F}\circ\mathcal{C}^{(k+1)}_{\text{noisy}}$.

\begin{prop}
If the noise channels acting on the class $\mathcal C$ of quantum circuits are $\mathcal{C}$-compatible, 
then 
\begin{eqnarray}
R_S\left(\mathcal{F}\circ\mathcal{C}^{(k+1)}_{\text{noisy}}\right)
\leq R_S\left(\mathcal{F}\circ\mathcal{C}^{(k)}_{\text{noisy}}\right),
\end{eqnarray}
for any integer $k\geq 0$.
\end{prop}
\begin{proof}
The proof of this statement is straightforward.
Since the noise channels are $\mathcal{C}$-compatible, 
each noisy quantum circuit in $\mathcal{C}^{(k+1)}_{\text{noisy}}$
expressed as a quantum channel can be written as a convex combination 
of the noisy quantum circuits in $\mathcal{C}^{(k)}_{\text{noisy}}$, i.e.~$
\mathcal{C}^{(k+1)}_{\text{noisy}}
\subset \Conv(\mathcal{C}^{(k)}_{\text{noisy}})
$.
Therefore, we have 
\begin{eqnarray*}
R_S\left(\mathcal{F}\circ\mathcal{C}^{(k+1)}_{\text{noisy}}\right)
\leq R_S\left(\mathcal{F}\circ \Conv(\mathcal{C}^{(k)}_{\text{noisy}})\right)
= R_S\left(\mathcal{F}\circ\mathcal{C}^{(k)}_{\text{noisy}}\right),
\end{eqnarray*}
where the equality comes from the fact that the Rademacher complexity is invariant under
convex combinations. 
\end{proof}

The above proposition tells us that 
the Rademacher complexity of quantum circuits under noise is nonincreasing in general. 

 \begin{figure}[!ht]
  \center{\includegraphics[width=8cm]  {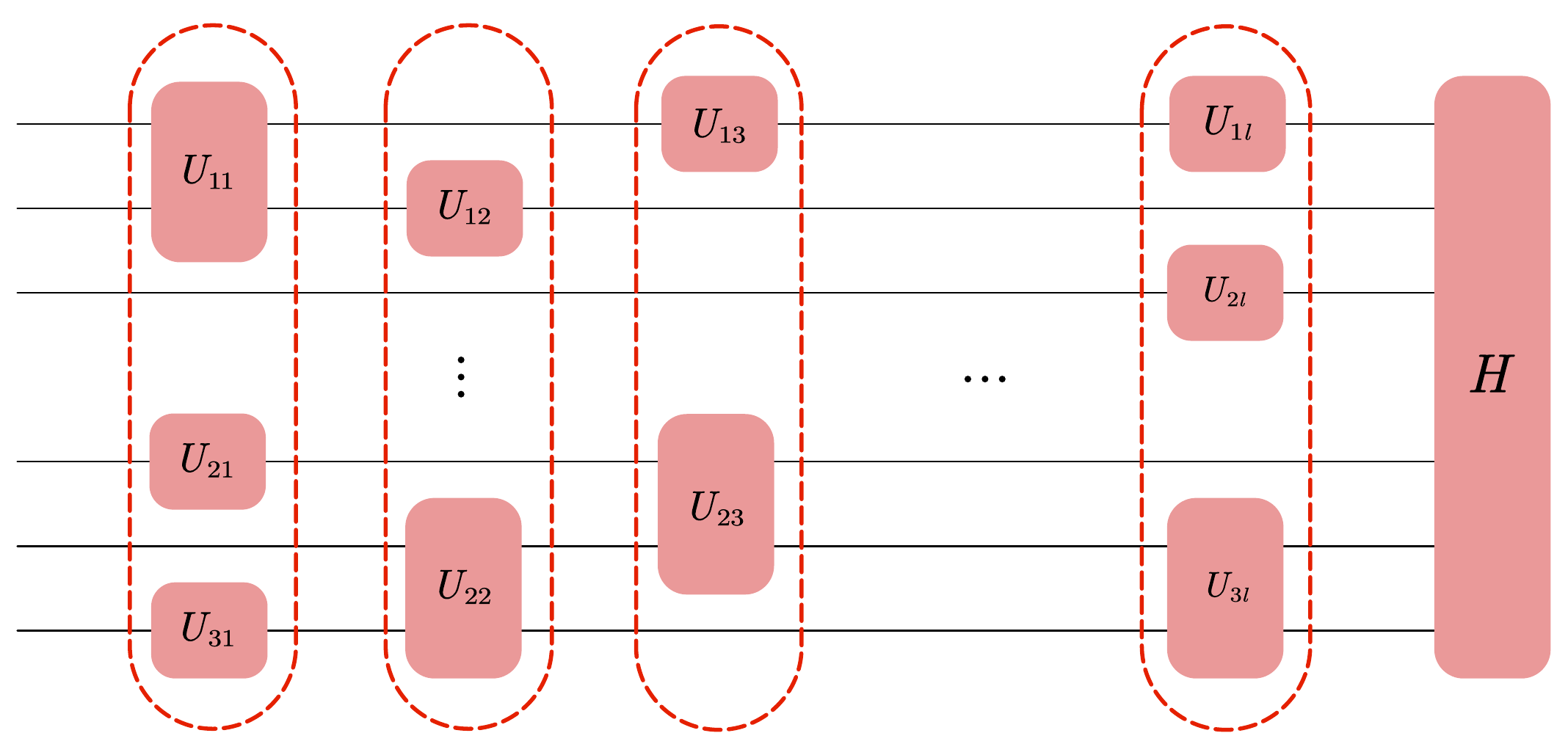}}     
  \caption{Circuit diagram of a quantum circuit with a fixed structure.}
  \label{fig1}
 \end{figure}

Next, let us consider the effects of noise on the set of quantum circuits with a fixed structure and bounded resources. To this end, we will make use of a resource measure of magic that we introduced in \cite{bu2021on} for quantum channels, namely the $(p,q)$ group norm $\norm{M^{\Phi}}_{p,q}$ of the representation matrix of quantum channels with respect to the Pauli basis, defined as
\begin{eqnarray}
M^{\Phi}_{\vec{z}\vec{x}}
=\frac{1}{2^{n_2}}\Tr{P_{\vec{z}}\Phi(P_{\vec{x}})},
\end{eqnarray} 
where $P_{\vec{x}}$, $P_{\vec{z}}$ are the Pauli operators corresponding to the strings $\vec{x}\in \set{0,1,2,3}^{n_1}$ and  $\vec{z}\in \set{0,1,2,3}^{n_2}$
. For any $N_1\times N_2$ matrix $M$, 
the $(p,q)$ \textit{group norm} of $M$, where $0<p,q\leq \infty$, is defined as
$
\norm{M}_{p,q}=
\left(\frac{1}{N_1}\sum_{i}\norm{M_i}^q_p\right)^{1/q}
$ and $\norm{M_i}_p=\left(\sum_{j}M^p_{ij}\right)^{1/p}$.

Given a set of quantum circuits with a fixed structure $\mathcal{A}$ (for example, see Figure \ref{fig1}), 
let us define  the vector $\vec{\mu}_{p,q}=\left(\norm{M^{\Phi_{ij}}}_{p,q}\right)_{ij}$, where $\Phi_{ij}$ is the 
$i$-th quantum channel in the $j$-th layer. 
Furthermore, let us define the inequality $\vec{\mu}_{p,q}\leq \vec{\mu}$ to mean
$\norm{M^{\Phi_{ij}}}_{p,q}\leq  \mu_{ij}$ for all $i,j$. We define
$
\mathcal{C}^{\mathcal{A}}_{\vec{\mu}_{p,q}\leq\vec{\mu} }
$ to be the set of quantum circuits with a fixed structure $\mathcal{A}$
with bounded resource $\vec{\mu}_{p,q}\leq \vec{\mu}$. Here, we consider
the case where $p=1, q=\infty$, i.e., $
\mathcal{C}^{\mathcal{A}}_{\vec{\mu}_{1,\infty}\leq\vec{\mu} }
$.

We now describe our noise model. Recall the definition of the single-qubit depolarizing channel 
$D_{\epsilon}$ defined in Eq.~\eqref{eq:depC}. 
We will consider noisy quantum circuits where  there is a (local) depolarizing channel acting after each gate
 $\Phi$ (See Figure \ref{fig2}). 
  \begin{figure}[!ht]
  \center{\includegraphics[width=8cm]  {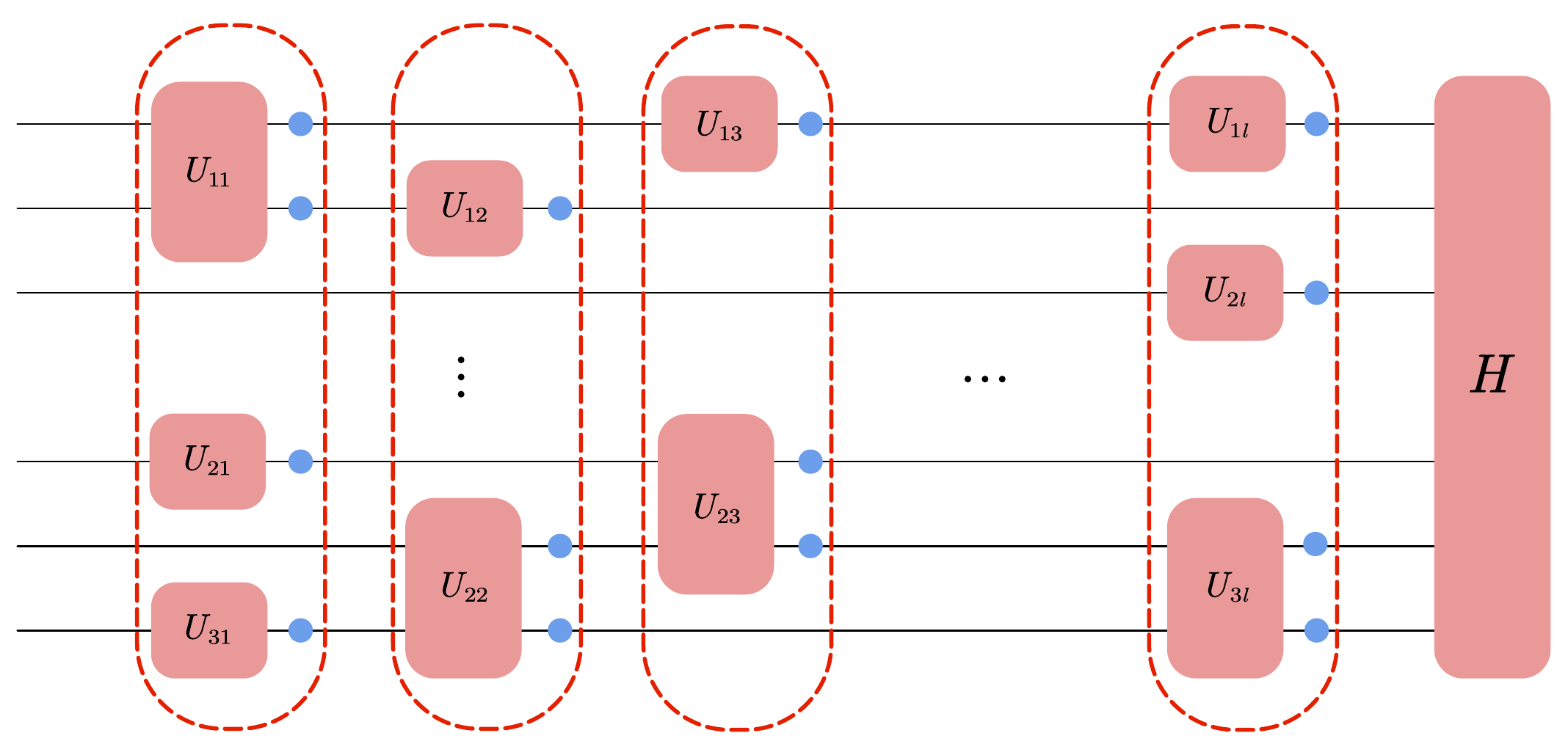}}     
  \caption{Circuit diagram of a noisy quantum circuit with a fixed structure, where the blue points denote single-qubit depolarizing channels.}
  \label{fig2}
 \end{figure}
  We denote such noisy channels as $\Phi_{\epsilon}=\ot D_{\epsilon}\circ \Phi$. To allow for different channels to have different noise parameters, we use $\epsilon_{ij}$ to denote the parameter in the depolarizing channel acting on the $i$-th
 gate of the $j$-th layer, i.e.~ $\Phi_{\epsilon_{ij}}=\ot  D_{\epsilon_{ij}}\circ \Phi_{ij}$.  
This gives us the following relationship between the Rademacher complexity of quantum circuits.

\begin{prop}\label{Prop:main2}
Given a set of quantum circuits $\mathcal{C}^{\mathcal{A}}_{\vec{\mu}_{1,\infty}\leq\vec{\mu} }$ with a fixed structure $\mathcal{A}$ and bounded resources,
and the noisy quantum circuits $\mathcal{C}^{\mathcal{A}, noisy}_{\vec{\mu}_{1,\infty}\leq\vec{\mu} }$, 
the Rademacher complexity on $m$ independent samples $S=\set{\vec{x}_1,\ldots,
\vec{x}_m}$ satisfies the following bound
\begin{align}
 R_S(\mathcal{F}\circ\mathcal{C}^{\mathcal{A}, noisy}_{\vec{\mu}_{1,\infty}\leq\vec{\mu} })
\leq R_S(\mathcal{F}\circ\mathcal{C}^{\mathcal{A}}_{\vec{\mu}_{1,\infty}\leq\vec{\mu}(\vec{\epsilon}) })
\end{align}
where the vector  $\vec{\mu}(\vec{\epsilon}) =(\mu_{ij}(\vec{\epsilon}))_{ij}$ is defined as
\begin{eqnarray}
\mu_{ij}(\vec{\epsilon})\leq 
(1-4\epsilon_{ij})\mu_{ij},
\end{eqnarray}
if $\max_{\Phi:\norm{M^{\Phi}}_{1,\infty}\leq \mu_{ij}}\norm{M^{\Phi_{\epsilon_{ij}}}}_{1,\infty}>1$.
Otherwise
\begin{eqnarray}
\mu_{ij}(\vec{\epsilon})\leq 1.
\end{eqnarray}

\end{prop}

We provide a proof of this proposition in  Appendix \ref{apen:pro1}. The following corollary follows directly 
from Proposition \ref{Prop:main2} and the results in \cite{bu2021on}.

\begin{cor}
The Rademacher complexity 
of  $\mathcal{C}^{\mathcal{A}, noisy}_{\vec{\mu}_{1,\infty}\leq\vec{\mu} }$ on $m$ independent samples $S=\set{\vec{x}_1,\ldots,
\vec{x}_m}$ satisfies the following bound

 \begin{eqnarray}
R_S(\mathcal{F}\circ\mathcal{C}^{\mathcal{A}, noisy}_{\vec{\mu}_{1,\infty}\leq\vec{\mu} })
\leq \prod_{ij}\mu_{ij}(\vec{\epsilon})\frac{\sqrt{8n_0}}{\sqrt{m}}\norm{\vec{\alpha}}_1\max_i\norm{\vec{f}_I(\vec{x}_i)}_{\infty},\nonumber\\
 \end{eqnarray}
 where 
 $\vec{\alpha}$ and $\vec{f}_I(\vec{x}_i)$  are the representation vectors of $H$ and $\proj{\psi(x_i)}$ in the Pauli basis, respectively. The representation vector $\vec{\alpha}^{Q}$ of a linear operator $Q$ in the Pauli basis is defined as 
 $\alpha^{Q}_{\vec{z}}=\frac{1}{2^n}\Tr{P_{\vec{z}}Q}$.
\end{cor}

\subsection{Lower bound}

\begin{prop}\label{prop:main3}
Consider $m$ independent samples $S=\set{\vec{z}_1,\ldots,\vec{z}_m}$. If the noisy channel is
$\mathcal{C}$-compatible,   then
the following relationship holds.
\begin{eqnarray}
R_S\left(\mathcal{F}\circ \mathcal{C}^{(k+1)}_{\text{noisy}}\right)
\geq \left(1-2\sum_jp_j\right)R_S\left(\mathcal{F}\circ\mathcal{C}^{(k)}_{\text{noisy}}\right).
\label{eq:lowerbound}
\end{eqnarray} 
\end{prop}
We provide a proof of this proposition in Appendix \ref{apen:main3}.
Moreover, we can get a better lower bound on the Rademacher complexity of 
noisy quantum circuits by introducing the free-robustness.

\begin{thm}\label{thm:main4}
Given $m$ independent samples $S=\set{\vec{z}_1,\ldots,\vec{z}_m}$, 
the following relationship holds
\begin{eqnarray}
R_S\left(\mathcal{F}\circ\mathcal{C}^{(k+1)}_{\text{noisy}}\right)
\geq \left(1+2\gamma_{k,k+1}\right)^{-1} R_S\left(\mathcal{F}\circ\mathcal{C}^{(k)}_{\text{noisy}}\right),
\end{eqnarray} 
where $\gamma_{k,k+1}$ is defined as

\begin{align}
\nonumber\gamma_{k,k+1}=\max_{C_{k}\in\mathcal{C}^{(k)}_{\text{noisy}} }\min\bigg\{\lambda &\bigg| \frac{C_{k}+\lambda C_{k+1}}{1+\lambda}\in \Conv(\mathcal{C}^{(k+1)}_{\text{noisy}}), \\
&\qquad \left.
C_{k+1}\in \Conv(\mathcal{C}^{(k+1)}_{\text{noisy}})\right\}.
\end{align}
\end{thm}
We provide a proof of this theorem in Appendix \ref{apen:main4}. If the noise channel has a recovery map $\mathcal{E}_R(\cdot)=\sum_iv_iU_i(\cdot) U^\dag_i$, i.e., 
$\mathcal{E}_R\circ D=id$, and each $U_i(\cdot)U_i^\dag$ is $\mathcal{C}$-compatible, then we can define the $l_1$ norm $\norm{\vec{v}(\mathcal{E}_R)}_1$ as follows
\begin{align}\label{eq:recov_dec}
\norm{\vec{v}(\mathcal{E}_R)}_1
=\min &\Big\{\sum_i|v_i| : \mathcal{E}_R(\cdot)=\sum_iv_iU_i(\cdot)U^\dag_i,
\nonumber\\
&\qquad \text{$U_i(\cdot)U_i^\dag$ is $ \mathcal{C}$-compatible}\Big\},
\end{align}
where the vector $\vec{v}$ is the representation vector in $\mathcal{C}$.
Note that $\norm{\vec{v}(\mathcal{E}_R)}_1$ can provide a lower bound for 
the robustness $\gamma_{k,k+1}$.

\begin{prop}\label{prop:rec1}
If the noisy channel has a recovery map $\mathcal{E}_R$, then we have 
\begin{align}
\gamma_{k,k+1}\leq (\norm{\vec{v}(\mathcal{E}_R)}_1-1)/2,
\end{align}
if $\norm{\vec{v}(\mathcal{E}_R)}_1<\infty$, which means that there exists a decomposition of the form stated in \eqref{eq:recov_dec}.

\end{prop}

\begin{proof}
By the definition of recovery map, we have 
$\mathcal{E}_{R}\circ D=\mathrm{id}$. Thus, for any quantum circuit $C_{k}\in\mathcal{C}^{(k)}_{\text{noisy}}$, there exists
a quantum circuit $C_{k+1}\in\mathcal{C}^{(k+1)}_{\text{noisy}}$ such that $C_{k}$ can be obtained from  $C_{k+1}$
by adding recovery maps $\mathcal{E}_R$ into  $C_{k+1}$, which implies that 
 $C_{k}$ can be written as a linear combination of the quantum circuits in $\mathcal{C}^{(k+1)}_{\text{noisy}}$.
 Therefore $1+2\gamma_{k,k+1}\leq \norm{\vec{v}(\mathcal{E}_R)}_1$.
\end{proof}

Hence, we obtain the following statement directly from Theorem \ref{thm:main4} and Proposition \ref{prop:rec1}.

\begin{prop}\label{prop:cor6}
Given $m$ independent samples $S=\set{\vec{z}_1,\ldots,\vec{z}_m}$, if the noisy channel has a recovery map $\mathcal{E}_R$, then we have 
the following relationship
\begin{align}
R_S\left(\mathcal{F}\circ\mathcal{C}^{(k+1)}_{\text{noisy}}\right)
\geq \norm{\vec{v}(\mathcal{E}_R)}^{-1}_1 R_S\left(\mathcal{F}\circ\mathcal{C}^{(k)}_{\text{noisy}}\right).
\end{align} 
\end{prop}

Now, let us give some examples of noise channels to compare the bounds given by Propositions \ref{prop:main3} and \ref{prop:cor6}.

\noindent\textit{\textbf{Example 1}} (Depolarizing channel)
For the depolarizing channel $D_\epsilon$ defined 
in \eqref{eq:depC}, we have 
$1-2\sum_jp_j=1-6\epsilon$.
The depolarizing channel has a recovery map $\mathcal{E}_R$  \cite{Temme2017,Takagi2020}, which can be written as
\begin{eqnarray*}
\mathcal{E}_R(\rho)
=\left(1+\frac{3\epsilon}{1-4\epsilon}\right)
\rho-\frac{\epsilon}{1-4\epsilon}(X\rho X+Y\rho Y+Z\rho Z).
\end{eqnarray*}
Hence
we have
\begin{align}
(1+2\gamma_{k,k+1})^{-1}
&\geq \norm{\vec{v}(\mathcal{E}_R)}^{-1}
\geq \left(\frac{1+2\epsilon}{1-4\epsilon}\right)^{-1}
\nonumber>1-6\epsilon.
\end{align}
This tells us that the bound from  Proposition \ref{prop:cor6}
is better than that from Proposition \ref{prop:main3}.

\noindent\textit{\textbf{Example 2}} (Dephasing channel)
For the dephasing channel $D^{P}_{\epsilon}$ defined
as 
\begin{eqnarray}
D^{P}_{\epsilon}(\rho)
=(1-\epsilon)\rho
+\epsilon Z\rho Z,
\end{eqnarray}
we have $1-2\sum_jp_j=1-2\epsilon$. 
Its recovery map \cite{Takagi2020} is given by
\begin{eqnarray}
\mathcal{E}^P_{R}(\rho)
=\frac{1-\epsilon}{1-2\epsilon}
\rho-\frac{\epsilon}{1-2\epsilon}Z\rho Z.
\end{eqnarray}
Hence $\norm{\vec{v}(\mathcal{E}_R)}^{-1}\geq 1-2\epsilon$, i.e.~the bound 
from Proposition \ref{prop:cor6}
is the same as that from Proposition \ref{prop:main3}.

Using Theorem \ref{thm:main4} and Proposition \ref{prop:cor6}, we can also provide a 
lower bound on the Rademacher complexity of noisy quantum circuits with fixed structure $A$
and bounded resource
$\mathcal{C}^{\mathcal{A}, noisy}_{\vec{\mu}_{1,\infty}\leq\vec{\mu} }$
that depends on that of the noiseless quantum circuits $\mathcal{C}^{\mathcal{A}}_{\vec{\mu}_{1,\infty}\leq\vec{\mu} }$
and noise parameters.
\begin{cor}
Given a set of quantum circuits $\mathcal{C}^{\mathcal{A}}_{\vec{\mu}_{1,\infty}\leq\vec{\mu} }$ with a fixed structure $\mathcal{A}$ and bounded resource, 
and the noisy quantum circuits $\mathcal{C}^{\mathcal{A}, noisy}_{\vec{\mu}_{1,\infty}\leq\vec{\mu} }$, 
the Rademacher complexity on $m$ independent samples $S=\set{\vec{x}_1,\ldots,
\vec{x}_m}$ satisfies the following bound

\begin{eqnarray}
\prod_{ij} \left(\frac{1-4\epsilon_{ij}}{1+2\epsilon_{ij}}\right)^{\ot n_{ij}} 
R_S(\mathcal{F}\circ\mathcal{C}^{\mathcal{A}}_{\vec{\mu}_{1,\infty}\leq\vec{\mu} })
\leq R_S(\mathcal{F}\circ\mathcal{C}^{\mathcal{A}, noisy}_{\vec{\mu}_{1,\infty}\leq\vec{\mu} }),\nonumber\\
\end{eqnarray}
where $n_{ij}$ is the number of output qubits of the quantum channel $\Phi_{ij}$.
\end{cor}

\section{Concluding remarks}
In this study, we investigated how noise affects the Rademacher complexity of quantum circuits by considering noise models represented by mixed-unitary channels and computing lower and upper bounds for the Rademacher complexities of quantum circuits described by these noise models. More specifically, we found a lower bound for the Rademacher complexity of noisy quantum circuits that depends on the Rademacher complexity of the corresponding noiseless quantum circuit and the free robustness of the circuit. From our results, we see that noise decreases the Rademacher complexity of quantum circuits.

\begin{acknowledgments}
K. B. thanks Arthur Jaffe  and Zhengwei Liu for the help and support during the breakout of the COVID-19 pandemic. 
K. B. acknowledges the support of
 ARO Grants W911NF-19-1-0302 and
W911NF-20-1-0082, and the support from Yau Mathematical
Science Center at Tsinghua University during the visit. 

\end{acknowledgments}

 \bibliography{SatCom-lit}

\appendix
\widetext

\section{Proof of Proposition \ref{Prop:main2}}\label{apen:pro1}

\begin{proof}
To prove this statement, it suffices to prove that 
if $\max_{\Phi:\norm{M^{\Phi}}_{1,\infty}\leq \mu}\norm{M^{\Phi_{\epsilon}}}_{1,\infty}>1$, then
\begin{eqnarray}
\max_{\Phi:\norm{M^{\Phi}}_{1,\infty}\leq \mu}\norm{M^{\Phi_{\epsilon}}}_{1,\infty}
\leq (1-4\epsilon)\mu.
\end{eqnarray}
By the definition of $\norm{\cdot}_{1,\infty}$, we have
\begin{eqnarray*}
\norm{M^{\Phi_{\epsilon}}}_{1,\infty}
&=&\max\left\{\norm{M^{\Phi_{\epsilon}}_{\vec{0}}}_1, \max_{\vec{z}\neq \vec{0}}
\norm{M^{\Phi_{\epsilon}}_{\vec{z}}}_1
\right\}\\
&=&\max\left\{\norm{M^{\Phi_{\epsilon}}_{\vec{0}}}_1, \max_{\vec{z}\neq \vec{0}}
(1-4\epsilon)^{w(\vec{z})}
\norm{M^{\Phi}_{\vec{z}}}_1
\right\}\\
&=&\max\left\{1, \max_{\vec{z}\neq \vec{0}}
(1-4\epsilon)^{w(\vec{z})}
\norm{M^{\Phi}_{\vec{z}}}_1
\right\}.
\end{eqnarray*}
Thus, 
if $\max_{\Phi:\norm{M^{\Phi}}_{1,\infty}\leq \mu}\norm{M^{\Phi_{\epsilon}}}_{1,\infty}>1$, 
then we have 
\begin{eqnarray*}
\max_{\Phi:\norm{M^{\Phi}}_{1,\infty}\leq \mu}\norm{M^{\Phi_{\epsilon}}}_{1,\infty}
&=&\max_{\Phi:\norm{M^{\Phi}}_{1,\infty}\leq \mu}
\max_{\vec{z}\neq \vec{0}} (1-4\epsilon)^{w(\vec{z})}\norm{M^{\Phi}_{\vec{z}}}_1\\
&\leq& (1-4\epsilon)\max_{\Phi:\norm{M^{\Phi}}_{1,\infty}\leq \mu}\max_{\vec{z}\neq \vec{0}} \norm{M^{\Phi}_{\vec{z}}}_1\\
&\leq&(1-4\epsilon)\max_{\Phi:\norm{M^{\Phi}}_{1,\infty}\leq \mu}\norm{M^{\Phi}}_{1,\infty}\\
&\leq& (1-4\epsilon)\mu.
\end{eqnarray*}

\end{proof}

\section{Proof of Proposition \ref{prop:main3}}\label{apen:main3}

\begin{proof}
Since the noisy channels can be written as 
$D(\rho)=
(1-\sum_{j}p_j)
\rho+\sum_jp_jU_j\rho U^\dag_j$, 
 any noisy quantum circuit $C\in \mathcal{C}^{(k+1)}_{\text{noisy}}$ can be written as
 \begin{eqnarray}
 C=(1-\sum_{j}p_j)
C_0+\sum_jp_jC_j,
 \end{eqnarray}
where each $C_i\in \mathcal{C}^{(k)}_{\text{noisy}} $. It follows that 
\begin{eqnarray}
f_C
=(1-\sum_{j}p_j)
f_{C_0}+\sum_jp_jf_{C_j}.
\end{eqnarray}
By the triangle inequality, 
\begin{eqnarray*}
\left|\sum^m_{i=1}\epsilon_i f_C(z_i)\right|
&\geq& (1-\sum_{j}p_j)\left|\sum^m_{i=1}\epsilon_i f_{C_0}(z_i)\right|
-\sum_jp_j\left|\sum^m_{i=1}\epsilon_i f_{C_j}(z_i)\right|\\
&\geq& (1-\sum_{j}p_j)\left|\sum^m_{i=1}\epsilon_i f_{C_0}(z_i)\right|
-\sum_jp_j \sup_{C_j\in \mathcal{C}^{(k)}_{\text{noisy}}}\left|\sum^m_{i=1}\epsilon_i f_{C_j}(z_i)\right|.
\end{eqnarray*}
Therefore, we have 
\begin{eqnarray}
R_S\left(\mathcal{F}\circ \mathcal{C}^{(k+1)}_{\text{noisy}}\right)
\geq \left(1-2\sum_jp_j\right)R_S\left(\mathcal{F}\circ\mathcal{C}^{(k)}_{\text{noisy}}\right).
\end{eqnarray}

\end{proof}

\section{Proof of Theorem \ref{thm:main4}}\label{apen:main4}

\begin{proof}
By the definition of $\gamma_{k,k+1}$,  for any noisy quantum circuit $C\in \mathcal{C}^{(k)}_{\text{noisy}} $,  there exist two 
noisy quantum circuits $C_1, C_2\in \Conv(\mathcal{C}^{(k+1)}_{\text{noisy}}) $ such that 
\begin{eqnarray}
C=(1+\gamma_{k,k+1})C_1-\gamma_{k,k+1}C_2.
\end{eqnarray}
Therefore, we have 
\begin{eqnarray}
\mathcal{C}^{(k)}_{\text{noisy}} 
\subset 
(1+\gamma_{k,k+1})\Conv(\mathcal{C}^{(k+1)}_{\text{noisy}} )-\gamma_{k,k+1}
Conv(\mathcal{C}^{(k+1)}_{\text{noisy}}).
\end{eqnarray}
Hence,
\begin{eqnarray}
\mathcal{F}\circ 
\mathcal{C}^{(k)}_{\text{noisy}} 
\subset
(1+\gamma_{k,k+1})\mathcal{F}\circ \Conv(\mathcal{C}^{(k+1)}_{\text{noisy}} )-\gamma_{k,k+1}
\mathcal{F}\circ \Conv(\mathcal{C}^{(k+1)}_{\text{noisy}}).
\end{eqnarray}
Therefore, we have 
\begin{eqnarray}
R_S\left(\mathcal{F}\circ 
\mathcal{C}^{(k)}_{\text{noisy}} \right)
\leq 
(1+2\gamma_{k,k+1})R_S\left(\mathcal{F}\circ \Conv(\mathcal{C}^{(k+1)}_{\text{noisy}} )\right)
=(1+2\gamma_{k,k+1})R_S\left(\mathcal{F}\circ \mathcal{C}^{(k+1)}_{\text{noisy}} \right),
\end{eqnarray}
where the last equality follows from the fact that the Rademacher complexity is invariant under
convex combinations.
\end{proof}

\end{document}